\documentclass[11pt]{article}%
\usepackage{amsfonts}
\usepackage{amsmath}
\usepackage{amssymb}
\usepackage{graphicx}
\usepackage{fullpage}
\usepackage{dsfont}
\usepackage{hyperref}%

\usepackage{dcolumn}
\usepackage{bm}
\usepackage[utf8]{inputenc}
\usepackage[T1]{fontenc}
\usepackage{mathptmx}
\usepackage{float}

\providecommand{\U}[1]{\protect\rule{.1in}{.1in}}
\newtheorem{theorem}{Theorem}

\newtheorem{lemma}[theorem]{Lemma}

\newtheorem{proposition}[theorem]{Proposition}
\newtheorem{remark}[theorem]{Remark}

\newenvironment{proof}[1][Proof]{\noindent\textbf{#1.} }{\ \rule{0.5em}{0.5em}}
\numberwithin{equation}{section}
\begin{document}

\title{\textbf{Classical capacities of memoryless but not identical quantum channels}}
\author{Samad Khabbazi Oskouei\thanks{Department of Mathematics, Varamin-Pishva Branch, Islamic Azad University, Varamin 33817-7489, Iran}
\and Stefano Mancini\thanks{School of Science and Technology, University of
Camerino, Via M.~delle Carceri 9, I-62032 Camerino, Italy \& INFN--Sezione
Perugia, Via A.~Pascoli, I-06123 Perugia, Italy}
}
\maketitle

\begin{abstract}
We study quantum channels that  vary on time in a deterministic way, that is, they change in an independent but not identical way from one to another use. We derive coding theorems for the classical entanglement assisted and unassisted capacities. We then specialize the theory to lossy bosonic quantum channels and show the existence of contrasting examples where capacities can or cannot be drawn from the limiting behavior of the lossy parameter.
\end{abstract}

\maketitle

\section{Introduction}

Any physical process involves a state change and hence can be regarded as a quantum channel,
i.e. a stochastic map on the set of density operators \cite{Mbook}. As such, it results quite naturally to characterize physical processes in terms of their ability  to transmit information.
Hence, much attention has been devoted to quantum channel capacities.
They were, however, mostly confined to the assumption of channels acting in independent and identically distributed way over inputs. Only recently it has been started to go beyond this assumption \cite{CGLM}.

A paradigm in this direction is provided by compound quantum channels where the map, though being the same over the uses, is initially randomly selected from a given set (with possibly infinite many elements)\cite{BB09}. A more general situation is represented by arbitrarily \textit{varying} quantum channels,
where the sender and receiver must deal with further uncertainty. In fact, in this case, the map is randomly chosen (from a given set) at each use.
Concerning this latter, wider class of quantum channels, the classical capacity was derived in \cite{AB07} under the assumption of classical-quantum channels and then extended to the fully quantum channels in \cite{BDW18}. Instead, the entanglement-assisted classical capacity has been derived in \cite{BHK16}.
In particular, in Ref.\cite{AB07} it was proved that the average error classical capacity of a classical-quantum arbitrarily varying channel equals zero or else the random code capacity. Conditions for the latter case were found by using the elimination (or de-randomization) technique. Then, Ref.\cite{BDW18} showed how the random code capacity of a finite-dimensional quantum arbitrarily varying channel can be reduced to the capacity of a naturally associated compound channel by using permutation symmetry and de Finetti reduction.
The technique used in Ref.\cite{BHK16} to prove a coding theorem for the entanglement-assisted classical capacity relies on the arguments used in Ref.\cite{ABBJ13} for finite-dimensional arbitrarily varying quantum channels. That is, capacity-achieving codes for general compound quantum channels were used. Next a variation of the so-called robustification and elimination techniques was borrowed from \cite{JOPS2012} as a method to extend the coding theorem to arbitrarily varying quantum channels.

Here we want to consider channels that are \textit{varying} from one use to another, not in an arbitrary (random) way, but rather in a deterministic way. As such they cannot however be obtained as a particular case of arbitrarily varying quantum channels, because even if we concentrate the probability measure used therein to one item of the channels' set, we recover the independent and identical distributed model. Instead, we would deal with independent but not identical channels and more specifically with classical information transmission through them.
We shall derive the classical assisted and unassisted capacities formulae in a different way with respect to Refs.\cite{AB07,BDW18,BHK16}. Namely, we shall employ position-based encoding \cite{QWW17,anshu2017one} and sequential decoding \cite{Wilde20130259}, which rely on quantum hypothesis testing and Berry-Esseen theorem [tools discussed in \cite{BD10,KW17a,li12,TH12} for finite-dimensional Hilbert spaces and in \cite{OSM2018} for infinite-dimensional ones]. As a
by-product, we obtain the validity of formulae also in the case of infinite-dimensional spaces. We then specialize the theory in this context by considering lossy bosonic quantum channels.
The study of deterministic time-varying quantum channels is motivated by the fact that a determinist description is often taken for linear time-varying channels in wireless communication (see e.g. \cite{Matz2011}), which is almost unexplored at the quantum level. Additionally, there are practical situations, of increasing interest for quantum communication, showing deterministically time-varying channels. One of these is provided by data transmission from a low-orbit satellite to a geostationary satellite or ground station. In such scenarios, the received signal power increases as the transmitting low-orbit satellite comes into view, and then decreases as it then departs, resulting in a communication link whose time variation is known to the sender and receiver (see e.g. \cite{Ryan1995}).

The paper is organized as  follows. Preliminary notions, starting from smooth quantum relative entropy, are introduced in Section \ref{sec:prel}. In the core part of the paper, we will derive coding theorems for the classical entanglement assisted and unassisted capacities (Sections \ref{sec:EACC} and \ref{sec:UCC} respectively). We then specialize the theory to lossy bosonic quantum channels (Section \ref{sec:Gauss}) and show the existence of contrasting examples where capacities can or cannot be drawn from the limiting behavior of the lossy parameter (Section \ref{sec:EX}).
Section \ref{sec:conclu} is for conclusions and Appendix \ref{ap:second}  contains
details on the deviation of smooth relative entropy
from standard relative entropy.


\section{Preliminaries}\label{sec:prel}

Let $\mathcal{H}$ be a separable Hilbert space, and let  $\mathcal{T}(\mathcal{H})$ be a set of
trace class linear operators acting on $\mathcal{H}$.
A quantum channel $\mathcal{N}_{A\rightarrow B}$ is a completely positive trace preserving (CPTP) linear map from $\mathcal{T}(\mathcal{H}_A)$ to $\mathcal{T}(\mathcal{H}_B)$.

Suppose that we have an infinite sequence $\mathfrak{N}=\{{\cal N}_k^{A\to B}\}_k$ of quantum channels, known to both the sender
and receiver before communication begins,  whence referred to as deterministic. Here we want to address the issue of what are the classical capacities (entangled assisted and unassisted) of such a sequence of
channels.
 To this end we cannot resort to the standard asymptotic theory that is valid for independent and identical channels.
Rather, what we will do is to consider the one-shot capacity (see e.g. \cite{ BD10,WR12, anshu2017one,  DH13, MW14}) for the first $n$-items of the channels sequence and then let $n$ goes to infinity. In other words, $n$ items of the sequence are viewed as a single (larger) channel which will be used only once and the number of bits that can be transmitted through it with a given average error probability will be found. In doing so we need to introduce tools like smooth quantum relative entropy\cite{BD10,WR12}.

\bigskip

A density operator on  $\mathcal{H}$ is a positive linear operator with trace equal to one.
Let us consider two density operators $\rho$ and $\sigma$ and
assume that their spectral decompositions are given by
\begin{equation}
\rho=\sum_{x\in\mathcal{X}} \lambda_x P_x \quad \text{and} \quad \sigma=\sum_{y\in\mathcal{Y}} \mu_y Q_y\,,
\end{equation}
where $\mathcal{X}$ and $\mathcal{Y}$ are countable index sets, $\{\lambda_x\}_{x\in\mathcal{X}}$ and $\{\mu_y\}_{y\in\mathcal{Y}}$ are probability distributions with $\sum_{x\in\mathcal{X}} \lambda_x=\sum_{y\in \mathcal{Y}}\mu_y = 1$, and $P_x, Q_y$ are projections such that $\sum_{x\in\mathcal{X}} P_x=\sum_{y\in \mathcal{Y}}Q_y=I$.

Given a Positive Operator Valued Measure
(POVM) with two elements, $\Pi$ and $I -\Pi$,
aimed at distinguishing $\rho$ from $\sigma$,
we consider a smoothed version of quantum relative entropy defined as
the negative logarithm of the minimum probability that
the `test' $\Pi$ will fail on
state $\sigma$, under the constraint that its failure probability
on state $\rho$ is not larger than $\varepsilon\in(0, 1)$, that is \cite{BD10,WR12}
 \begin{equation}
   D_H^\varepsilon(\rho \Vert \sigma)\equiv\sup_{0\leq \Pi \leq I,\, \operatorname{Tr}(\Pi \rho)\geq 1-\varepsilon} \left[
  -\log \operatorname{Tr}(\Pi \sigma) \right]\,.
 \end{equation}
Throughout this paper $\log$ stands for $\log_2$.

The quantum relative entropy \cite{L73}, its variance \cite{TH12, li12, KW17a} and the $T$ quantity
are respectively defined as \cite{TH12, li12, KW17a}:
\begin{equation}
D(\rho\Vert\sigma)\equiv\sum_{x\in\mathcal{X},y\in\mathcal{Y}}\lambda_x \operatorname{Tr}(P_x Q_y)\log\left(\frac{\lambda_x}{\mu_y}\right)\,,
\end{equation}
\begin{equation}
V(\rho\Vert\sigma)\equiv\sum_{x\in\mathcal{X},y\in\mathcal{Y}}\lambda_x \operatorname{Tr}(P_x Q_y)\left(  \log
\left(\frac{\lambda_x}{\mu_y}\right)-D(\rho\Vert\sigma)\right)  ^{2}\,,
\end{equation}
\begin{equation}\label{eq:Tdef}
T(\rho\Vert\sigma)\equiv\sum_{x\in\mathcal{X},y\in\mathcal{Y}}\lambda_x \operatorname{Tr}(P_x Q_y) \left\vert\log\left(\frac{\lambda_x}{\mu_y}\right)-D(\rho\Vert\sigma)\right\vert ^{3}\,.
\end{equation}
For given density operators $\rho$ and $\sigma$ satisfying
\begin{equation}
D(\rho\Vert \sigma), V(\rho\vert\sigma)\,\, \text{and}\,\, T(\rho\Vert\sigma) < \infty\,,
\end{equation}
we have the following expansion
\begin{equation}\label{re:7}
D_H^\varepsilon\left( \rho^{\otimes n}\Vert \sigma^{\otimes n}  \right)=n D(\rho\Vert\sigma)+\sqrt{n V(\rho\Vert\sigma)}\Phi^{-1}(\varepsilon)+O(\log n)\,,
\end{equation}
where
\begin{equation}
\Phi(a)\equiv\frac{1}{\sqrt{2\pi}}\int_{-\infty}^a  \operatorname{exp}\left(-\frac{x^2}{2}\right)dx,
\qquad
\Phi^{-1}(\varepsilon)=\sup\{ a\in \mathds{R} \vert \Phi(a)<\varepsilon\}\,.
\end{equation}
Relation ~\eqref{re:7} was proven for finite dimensional Hilbert spaces in ~\cite{TH12, li12}.
 For infinite-dimensional separable Hilbert spaces, the inequality $\leq$ was proven in \cite{DPR15,KW17a}, while the inequality $\geq$ was shown in \cite{OSM2018}.
 In Appendix \ref{ap:second}, we generalize it as follows
\begin{align}\label{newDDrel}
D_{H}^{\varepsilon}\left( \bigotimes_{i=1}^n\rho_i \Bigg\Vert\bigotimes_{i=1}^n\sigma_i\right)  =\sum_{i=1}^n D(\rho_i\Vert\sigma_i)+\sqrt{\sum_{i=1}^n V(\rho_i\Vert\sigma_i)}\,\Phi^{-1}\!\left(
\varepsilon\right)  +O(\log n)\,,
\end{align}
where $\rho_i$'s and $\sigma_i$'s are density operators acting on $\mathcal{H}$ with the additional condition
 \begin{equation}
  \lim_{n\to \infty}\frac{6 \sum_{i=1}^n \left(T(\rho_i \| \sigma_i)\right)}{\sqrt{\left(\sum_{i=1}^n V\left(\rho_i \Vert \sigma_i\right)\right)^3}}=0\, .
\end{equation}


\section{Entanglement assisted classical capacity}\label{sec:EACC}

In this Section we present the coding theorem for the entanglement assisted classical capacity
of a deterministic sequence of independent channels $\mathfrak{N}=\{{\cal N}_k^{A\to B}\}_k$.

Suppose that  channels $\mathfrak{N}$ connect a sender
Alice to a receiver Bob
and they can share an arbitrary quantum state $\rho_{R^nA^n}$ before using the first $n$
items of $\mathfrak{N}$.
Here the system $R$ (resp. $A$) is considered as accessible to Bob (resp.) Alice.
For positive integers $n$
and $M$, and $\varepsilon\in\left[  0,1\right]  $, an $(n,M,\varepsilon)$ code
for entanglement-assisted classical communication consists of the resource
state $\rho_{R^nA^n}$ and a set $\{\mathcal{E}_{A^n\rightarrow A^{n}}^{m}\}_{m\in\mathcal{M}}$ of encoding channels, where
$\left\vert \mathcal{M}\right\vert =M$. It also consists of a decoding POVM
$\left\{  \Lambda_{R^{n}B^n}^{m}\right\}  _{m\in\mathcal{M}}$
satisfying the following condition:
\begin{equation}
\frac{1}{M}\sum_{m\in\mathcal{M}}\operatorname{Tr}\{\left(  I_{R^{n}B^n}-\Lambda_{R^{n}B^{n}}^{m}\right)  \otimes_{k=1}^n\mathcal{N}_k^{A\to B}(\mathcal{E}_{A^{n}\rightarrow A^{n}}^{m}
(\rho_{R^{n}A^{n}}))\}\leq\varepsilon\, ,
\end{equation}
which we interpret as saying that the average error probability is no larger
than $\varepsilon$, when using the entanglement-assisted code described above.

The entanglement-assisted classical capacity of the first $n$ items of $\mathfrak{N}$, denoted by $C_{{E}}(\mathfrak{N}%
,n,\varepsilon)$, is equal to the largest value of $\frac
{1}{n}\log M$ (bits per channel use) for which there exists an $\left(
n,M,\varepsilon\right)  $ entanglement-assisted code as described above.
The entanglement assisted classical capacity for $\mathfrak{N}$ is defined by
\begin{equation}
  C_{{E}}(\mathfrak{N})\equiv\lim_{\varepsilon\to 0}\lim_{n\to\infty} C_{{E}}(\mathfrak{N},n,\varepsilon)\,.
\end{equation}

\begin{theorem}\label{CEAcoding}
Given a deterministic sequence of independent channels $\mathfrak{N}=\{{\cal N}_k^{A\to B}\}_k$, the entanglement assisted classical capacity results
\begin{equation}\label{eq:CEN}
{C_{E}}(\mathfrak{N})=\lim_{n\to\infty}\frac{1}{n}\left[\max_{\rho_{RA}}
\sum_{k=1}^n D\left({\cal N}_k^{A\to B}(\rho_{RA})
\Big\|
\rho_R\otimes {\cal N}_k^{A\to B}(\rho_{A})\right) \right],
\end{equation}
where $\rho_{RA}$ is a resource entangled state shared by Alice and Bob.
\end{theorem}

To prove the Theorem, we will resort to position-based encoding and sequential decoding strategy. In other words, Alice and Bob are supposed to share $M$ resource entangled states $\rho_{R_i A_i}$, $i=1, \cdots, M$ where Bob has $R$ systems and Alice has $A$ systems. If Alice wants to transmit the message $m$ through the channel $\mathcal{N}^{A\to B}$, simply selects the $m$'s state in her systems and sends it through the channel so that the marginal state of Bob systems is as follows
\begin{equation}
\rho_{R_{1}}\otimes\cdots\otimes\rho_{R_{m-1}}\otimes
\mathcal{N}^{A_{m}\rightarrow B} (\rho_{R_{m}A_{m}})\otimes\rho_{R_{m+1}}\otimes
\cdots\otimes\rho_{R_{M}}\,.
\end{equation}
Bob then to determine which message Alice
transmitted, introduces $M$ auxiliary probe systems in the state
$|0\rangle\langle0|$, so that his overall state is
\begin{equation}
\omega_{R^{M}B P^{M}}^{m}\equiv\rho_{R_{1}}\otimes\cdots\otimes\rho_{R_{m-1}}\otimes\mathcal{N}^{A_{m}\rightarrow B} (\rho_{R_{m}A_{m}})\otimes
\rho_{R_{m+1}}\otimes\cdots\otimes\rho_{R_{M}}\otimes|0\rangle\langle
0|_{P_{1}}\otimes\cdots\otimes|0\rangle\langle0|_{P_{M}}\,.
\end{equation}
He next performs the binary measurements
 $\{\Pi_{R_{i}B_m P_{i}},\hat{\Pi}_{R_{i}B_m P_{i}}
\equiv I_{R_{i}B_m P_{i}}-\Pi_{R_{i}B_m P_{i}}\}$,
sequentially, in the order $i=1$, $i=2$, etc. With this
strategy, the probability that he decodes the $m$th
message correctly is given by
\begin{equation}
\operatorname{Tr}\{\Pi_{R_{m}B P_{m}}\hat{\Pi}_{R_{m-1}B P_{m-1}}\cdots\hat{\Pi
}_{R_{1}B P_{1}}\omega_{R^{M}B P^{M}}^{m}\hat{\Pi}_{R_{1}B P_{1}}\cdots\hat{\Pi
}_{R_{m-1}B P_{m-1}}\}\,.
\end{equation}
Applying the ``quantum union bound" \cite{OSM2018}, we can bound the complementary
probability (error probability)
\begin{align}
p_{\text{e}}(m) &  \equiv1-\operatorname{Tr}\{\Pi_{R_{m}BP_{m}}\hat{\Pi
}_{R_{m-1}BP_{m-1}}\cdots\hat{\Pi}_{R_{1}BP_{1}}\omega_{R^{M}BP^{M}}^{m}%
\hat{\Pi}_{R_{1}BP_{1}}\cdots\hat{\Pi}_{R_{m-1}BP_{m-1}}\}\,.
\end{align}
More specifically, by \cite[Theorem 5.1]{OSM2018} $p_{\text{e}}(m)\leq \varepsilon$ holds  for all $m$, when
\begin{equation}\label{pebound}
\log M=D_{H}^{\varepsilon-\eta}\left(\mathcal{N}(\rho_{RB})\Vert \rho_R\otimes\mathcal{N}(\rho_A) \right)-\log(4\varepsilon/\eta^{2}),
\end{equation}
where $\eta\in(0, \varepsilon)$ and $\varepsilon\in(0,1)$.
Now, we are going to use this relation to provide a lower bound on the position-based encoding and sequential decoding for the entangled assisted classical capacity of
the channel sequence $\{\mathcal{N}_k^{A\to B}\}_k$.

\begin{lemma}\label{CEALB}
A message $m\in{\cal M}$ can be sent through the channels $\bigotimes_{k=1}^n \mathcal{N}_k$
with $p_{\text{e}}(m)\leq \varepsilon$ by choosing
\begin{align}\label{cor:ea}
  \log M &=\sum_{k=1}^n D\left({\cal N}^{A\to B}_k(\rho_{R A}) \Big\Vert \rho_{R} \otimes
{\cal N}^{A\to B}_k(\rho_{A})\right) \notag\\
&+\sqrt{\sum_{k=1}^nV\left({\cal N}^{A\to B}_k(\rho_{R A}) \big\|
\rho_{R} \otimes {\cal N}^{A\to B}_k(\rho_{A})\right)}\Phi^{-1}\left(\varepsilon-\frac{1}{\sqrt{n}}\right)
+ O(\log n)\,.
\end{align}
\end{lemma}

\begin{proof}
By replacing $\mathcal{N}^{A\to B}$ with $\bigotimes_{k=1}^n\mathcal{N}_k^{A\to B}$ and
letting $\Lambda_{R_i^n B^n}^n$ be a measurement operator such that
\begin{equation}
  \Lambda_{R_i^n B^n}^n= \operatorname{argmax}_\Lambda \left( D_H^\varepsilon \left( \bigotimes_{k=1}^n \mathcal{N}^{A_i\rightarrow B}_k(\rho_{R_i A_i})
  \Big\Vert
  \bigotimes_{k=1}^n \rho_{R_i} \otimes
  \mathcal{N}^{A_i\rightarrow B}_k(\rho_{A_i})   \right)  \right)\,,
\end{equation}
we can get from \eqref{pebound}
\begin{equation}\label{re:19}
   \log{M}= {D_H^{\varepsilon-\eta}\left(\bigotimes_{k=1}^n \mathcal{N}^{A\rightarrow B}_k(\rho_{R A}) \Big\Vert \bigotimes_{k=1}^n \rho_{R}\otimes  \mathcal{N}^{A\rightarrow B}_k(\rho_{A})\right)}-\log\left( \frac{4\varepsilon}{\eta^2} \right)\,.
\end{equation}
Next, setting $\eta=1/\sqrt{n}$, with the second order asymptotic relation \eqref{newDDrel}
we arrive at
\begin{align}\label{{re:22}}
\log M &= \sum_{k=1}^n D\left({\cal N}^{A\to B}_k(\rho_{R A}) \big\Vert \rho_{R} \otimes
{\cal N}^{A\to B}_k(\rho_{A})\right) +\nonumber \\
&+\sqrt{\sum_{k=1}^nV\left({\cal N}^{A\to B}_k(\rho_{R A}) \| \rho_{R} \otimes
{\cal N}^{A\to B}_k(\rho_{A})\right)}\Phi^{-1}\left(\varepsilon-\frac{1}{\sqrt{n}}\right) + O(\log n)\,,
\end{align}
with the condition
\begin{equation}\label{re:Berry-condition}
  \lim_{n\to \infty}\frac{6 \sum_{k=1}^n \left[T\left({\cal N}^{A\to B}_k(\rho_{R A}) \Vert \rho_{R} \otimes {\cal N}^{A\to B}_k(\rho_{A})\right)\right]}{\sqrt{\left[\sum_{k=1}^nV\left(
  \mathcal{N}^{A\rightarrow B}_k(\rho_{R A}) \| \rho_{R} \otimes
  \mathcal{N}^{A\rightarrow B}_k(\rho_{A})\right)\right]^3}}=0\, .
\end{equation}
\hfill
\end{proof}

\begin{proof}\textbf{Theorem \ref{CEAcoding}.}
The direct part is based on the result of Lemma \ref{CEALB} which provides a lower bound
for the capacity, namely
\begin{align}\label{eq:lowB}
 {n\,C_{E}}(\mathfrak{N}, n, \varepsilon)&\geq \max_{\rho_{RA}}
\sum_{k=1}^n D\left({\cal N}^{A\to B}_k(\rho_{R A}) \big\Vert \rho_{R} \otimes
{\cal N}^{A\to B}_k(\rho_{A})\right) \notag\\
&+\sqrt{\sum_{k=1}^nV\left({\cal N}^{A\to B}_k(\rho_{R A}) \| \rho_{R} \otimes
{\cal N}^{A\to B}_k(\rho_{A})\right)}\Phi^{-1}\left(\varepsilon-\frac{1}{\sqrt{n}}\right) + O(\log n)\,.
\end{align}
For the converse part, since the channels are independent, though not identical, it is like to
have them used in parallel, hence
\begin{equation}\label{eq:uppB}
 {n\,C_{E}}(\mathfrak{N}, n, \varepsilon)\leq \max_{\rho_{RA}}
\sum_{k=1}^n D\left({\cal N}_k^{A\to B}(\rho_{RA})
\big\|\rho_R\otimes {\cal N}_k^{A\to B}(\rho_{A})\right).
\end{equation}
\hfill
\end{proof}

\begin{remark}
From Eqs.\eqref{eq:lowB} and \eqref{eq:uppB} it is evident that, by taking the limit $n\to\infty$,
the dependence on $\epsilon$ disappears. Hence the subsequent limit $\epsilon\to 0$ results superfluous.
\end{remark}

\begin{remark}
For arbitrarily varying quantum channel, Ref. \cite{BHK16}, the encoding depends on the dimension of input Hilbert space and so the capacity formula cannot work in the infinite dimensional case, while here we do not have such a restriction.
\end{remark}


\section{Unassisted classical capacity}\label{sec:UCC}

This Section is devoted to the coding theorem for the unassisted classical capacity of
a deterministic sequence of independent channel $\mathfrak{N}=\{{\cal N}_k^{A\to B}\}_k$.

Suppose that channels
$\mathfrak{N}$ connect a sender Alice to a receiver Bob. For
positive integers $n$ and $M$, and $\varepsilon\in\left[  0,1\right]  $, an
$(n,M,\varepsilon)$ code for classical communication consists of a set
$\{\rho_{A^{n}}^{m}\}$ of separable quantum states across systems $A_k$, which are called
quantum codewords, and where $\left\vert \mathcal{M}\right\vert =M$. It also
consists of a decoding POVM $\left\{  \Lambda_{B^{n}}^{m}\right\}
_{m\in\mathcal{M}}$ satisfying the following condition:
\begin{equation}\label{eq:avepe}
\frac{1}{M}\sum_{m\in\mathcal{M}}\operatorname{Tr}\{\left(  I_{B^{n}}
-\Lambda_{B^{n}}^{m}\right)   \otimes_{k=1}^n\mathcal{N}_k^{A\to B}
(\rho_{A^{n}}^{m})\}\leq\varepsilon\, ,
\end{equation}
which we interpret as saying that the average error probability is no larger
than $\varepsilon$, when using the quantum codewords and decoding POVM
described above.

The unassisted classical capacity with separable inputs of the first $n$ items of $\mathfrak{N}$, denoted by $C(\mathfrak{N},n,\varepsilon)$,
is equal to the largest value of $\frac{1}{n}\log M$ (bits per channel use)
for which there exists an $\left(  n,M,\varepsilon\right)  $ code as described above.
The unassisted classical capacity for $\mathfrak{N}$ is defined by
\begin{equation}
  C(\mathfrak{N})\equiv\lim_{\varepsilon\to 0}\lim_{n\to\infty} C(\mathfrak{N},n,\varepsilon)\, .
\end{equation}

 \begin{theorem}\label{Ccoding}
Given a deterministic sequence of independent channel $\mathfrak{N}=\{{\cal N}_k^{A\to B}\}_k$, the unassisted classical capacity  with separable inputs results
\begin{equation}\label{classcap}
C(\mathfrak{N})=\lim_{n\to\infty}\frac{1}{n}\left[\max_{\rho_{X^nA^n}}
\sum_{k=1}^n D\left({\cal N}_k^{A\to B}(\rho_{XA_k})\Big\|\rho_X\otimes {\cal N}_k^{A\to B}(\rho_{A_k})\right) \right]\, ,
\end{equation}
where
\begin{equation}
\rho_{X^nA^n}=\sum_{x^n\in{\cal X}^n}p(x^n)\, |x^n\rangle\langle x^n|
\otimes \rho_{A^n}^{x^n}
\end{equation}
is a classical-quantum state, with $|x^n\rangle\in{\cal H}_X^{\otimes n}$ orthonormal states
(being ${\cal H}_X$ a separable Hilbert space) and
$\rho_{A^n}^{x^n}=\left(\rho_{A_1}^{x_1} \otimes \ldots \otimes \rho^{x_n}_{A_n}\right)$.
\end{theorem}

Given a classical-quantum channel $x\rightarrow \rho_B^x$, we know
from \cite{wilde2017position} that there exists an encoding and position-based decoding for choosing
\begin{equation}\label{classq}
 \log M= D_H^{\varepsilon-\eta}\left( \rho_{X B}\Vert (\rho_X \otimes \rho_{B})\right)
 - \log(4\varepsilon/\eta^2),
 \end{equation}
where $\eta\in(0, \varepsilon)$, $\varepsilon\in(0,1)$ and
\begin{equation}
\rho_{XB}=\sum_x p(x)\vert x\rangle\langle x\vert\otimes  \rho_B^x\,,
\end{equation}
with $\rho_B^x\equiv{\cal N}(\rho_A^x)$.
 In other words, there exist an encoding $m\rightarrow \rho_A^{x_m}$ and position-based POVM
  $\{\Lambda_B^{x_m}\}_{m=1}^M$ as decoder such that, according to \eqref{eq:avepe}, and together with \cite[Theorem 5.1]{OSM2018}, $p_{\text{e}}(m)\leq \varepsilon$ holds  for all $m\in{\cal M}$.

\begin{lemma}\label{CLB}
A message $m\in {\cal M}$ can be sent through the channels $\bigotimes_{k=1}^n \mathcal{N}_k$
with $p_{\text{e}}(m)\leq \varepsilon$ by choosing
\begin{align}
  \log M &=\sum_{k=1}^n D\left({\cal N}^{A\to B}_k(\rho_{R A}) \big\Vert \rho_{R} \otimes
{\cal N}^{A\to B}_k(\rho_{A})\right) \notag\\
&+\sqrt{\sum_{k=1}^nV\left({\cal N}^{A\to B}_k(\rho_{R A}) \| \rho_{R} \otimes
{\cal N}^{A\to B}_k(\rho_{A})\right)}\Phi^{-1}\left(\varepsilon-\frac{1}{\sqrt{n}}\right) + O(\log n)\,.
\end{align}
\end{lemma}
\begin{proof}
If we replace $\rho_{XB}$ with $\bigotimes_{i=1}^n \rho_{XB_i}$ in Eq.\eqref{classq}, then we have
\begin{equation}\label{32}
  \log M { =}  D_H^{\varepsilon-\eta}\left(\bigotimes_{k=1}^n \rho_{X B_k}
  \Big\Vert \bigotimes_{k=1}^n(\rho_X \otimes \rho_{B_k})\right) -\log(4\varepsilon/\eta^2) \,.
\end{equation}
 As a consequence, following the arguments of Lemma \ref{CEALB},
we can get
 \begin{equation}\label{30}
\log M(\varepsilon) { =} \sum_{k=1}^n D\left(\rho_{X B_k} \Vert \rho_{X} \otimes
\rho_{B_k}\right) +\sqrt{\sum_{k=1}^nV\left(\rho_{X B_k} \| \rho_{R} \otimes \rho_{B_k}\right)}\Phi^{-1}\left(\varepsilon-\frac{1}{\sqrt{n}}\right) + O(\log n)\,,
\end{equation}
with the condition
\begin{equation}\label{Berrycond}
  \lim_{n\to \infty}\frac{6 \sum_{k=1}^n \left(T\left(\rho_{X B_k} \Vert \rho_{X} \otimes
\rho_{B_k}\right)\right)}{\sqrt{\left(\sum_{k=1}^nV\left(\rho_{X B_k} \Vert  \rho_{X} \otimes \rho_{B_k}\right)\right)^3}}=0\, .
\end{equation}
\hfill
\end{proof}

\medskip

\begin{proof}\textbf{Theorem \ref{Ccoding}.}
The direct part is based on the result of Lemma \ref{CLB} which provides a lower bound
for the capacity, namely
\begin{align}\label{lowB2}
 n\,C(\mathfrak{N},n,\varepsilon)\geq
\max_{\rho_{XA^n}}
\left[\sum_{k=1}^n D\left(\rho_{X B_k} \Vert \rho_{X} \otimes
\rho_{B_k}\right) +\sqrt{\sum_{k=1}^nV\left(\rho_{X B_k} \| \rho_{R} \otimes \rho_{B_k}\right)}\Phi^{-1}\left(\varepsilon-\frac{1}{\sqrt{n}}\right) + O(\log n)\right]\,.
\end{align}
For the converse part, since the channels are independent, though not identical, it is like to
have them used in parallel, hence  from the Holevo bound we have
\begin{align}\label{uppB2}
 n\,C(\mathfrak{N},n,\varepsilon) \leq \max_{\rho_{XA^n}}
\sum_{k=1}^n D\left(\rho_{X B_k} \Vert \rho_{X} \otimes
\rho_{B_k}\right).
\end{align}
\hfill
\end{proof}

\begin{remark}
From Eqs.\eqref{lowB2} and \eqref{uppB2} it is evident that, by taking the limit $n\to\infty$,
the dependence on $\epsilon$ disappears. Hence the subsequent limit $\epsilon\to 0$ results superfluous.
\end{remark}

\begin{remark}
We have employed the hypothesis testing based on the Berry-Esseen theorem to get the formula \eqref{30} as average of relative entropies of single channels. In contrast for arbitrarily varying quantum channels, Ref.\cite{AB07}, it was used the sum of $n$ entropies greater than $n$ times the smallest of them. Furthermore, in case we consider a finite dimensional Hilbert space,
and parallel the error bound \eqref{30} with that obtained for arbitrarily varying quantum channels \cite{AB07}, we note that the former only depends on the number of channel uses, while the latter also on the size of channels' set.
      \end{remark}
      

\section{Memoryless but not identical Gaussian lossy channels}\label{sec:Gauss}

Since the coding Theorems \ref{CEAcoding} and \ref{Ccoding} were derived without any restriction on the dimensionality of Hilbert spaces, they can be straightforwardly applied to continuous variable (bosonic) quantum channels.

We shall focus on a sequence of Gaussian lossy channels (each acting on a single bosonic mode)
$\{{\cal N}_{\eta_k}\}_{k=1}^\infty$, where
$\eta_k\in(0,1)$ is the trasmissivity characterizing the $k$th channel. As customary we shall also consider an average energy $N$ per channel use, so to have the constraint
\begin{equation}\label{constraint}
\sum_{k=1}^n N_k=n N\,,
\end{equation}
on the effective energy $N_k$ employed at $k$th use.


\subsection{Entangled assisted classical capacity}

We consider Alice and Bob sharing $M$ two-mode squeezed state each with photon mean number $N_k$. We want to see how the capacity resulting from Theorem \ref{CEAcoding} is approached over channel uses.

In Ref.\cite{HSH98} it has been shown that
\begin{equation}\label{QRE}
D\left({\cal N}_{A\to B}^k(\rho_{R A}) \| \rho_{R} \otimes {\cal N}_{A\rightarrow B}^k(\rho_{A})\right)=g(N_{k})+g(\eta_k N_{k})-g((1-\eta_k) N_{k})\,,
\end{equation}
where $g(x)\equiv(x+1)\log(x+1)-x\log x$.
In addition, the  quantum relative entropy variance is computed in \cite{KW17a} as
\begin{align}\label{re:34}
V\left({\cal N}_{A\to B}^k(\rho_{R A}) \| \rho_{R} \otimes {\cal N}_{A\to B}^k(\rho_{A})\right)
&=(1-\eta_k)N_{k} \left((1-\eta_k)N_{k}+1   \right)  \left[ \log\left( 1+\frac{1}{(1-\eta_k)N_{k}}   \right)   \right]^2 \notag\\
&-2(1-\eta_k)N_{k} (N_{k}+1) \log\left( 1+ \frac{1}{(1-\eta_k)N_{k}}    \right)\log\left( 1+ \frac{1}{N_{k}}   \right) \notag\\
&+ N_{k} (N_{k}+1 )\left[ \log\left(1+ \frac{1}{N_{k}}\right)     \right]^2 \, .
\end{align}
According to Theorem \ref{CEAcoding} and \eqref{QRE} we now need to
maximize the quantity
\begin{equation}\label{sumgV}
\sum_{k=1}^n \left[g\left(N_k\right)
+g\left(\eta_k N_k\right)-g\left((1-\eta_k)N_k\right)
\right],
\end{equation}
with respect to $N_k$. This amounts to set
\begin{align}
\delta \left\{ \text{Eq.\eqref{sumgV}} \right\}
=\sum_{k=1}^n \Bigg[& g'\left(N_k\right)
+\eta_k g'\left(\eta_k N_k\right)- (1-\eta_k) g'\left(({1-}\-\eta_k)N_k\right)\Bigg]
\delta N_k=0,
\end{align}
where $g'$ stands for the derivative of $g$ with respect to its argument.
From the energy constraint \eqref{constraint} we further have
\begin{equation}
\delta\left\{\sum_{k=1}^n N_k\right\}=\sum_{k=1}^n\delta N_k=0.
\end{equation}
Using a Lagrange multiplier $\beta$ we get
\begin{align}
\sum_{k=1}^n \Bigg[ g'\left(N_k \right)
+\eta_k g'\left(\eta_k N_k\right)- (1-\eta_k) g'\left(({1-}\-\eta_k)N_k\right)-\beta\Bigg] \delta N_k=0.
\end{align}
Solving the set of $n+1$ equations
\begin{equation}
\left\{
\begin{array}{cl}
g'\left(N_k \right)
+\eta_k g'\left(\eta_k N_k\right)- (1-\eta_k) g'\left(({1-}\-\eta_k)N_k\right)-\beta &=0 \\ \\
\sum_{k=1}^n N_k&=nN
\end{array}\right.,
\end{equation}
allows us to find the $N_k$ and $\beta$ giving
\begin{equation}\label{CEn}
 n \, \overline{C_{E}}(\{{\cal N}_{\eta_k}\},n)\equiv\max_{N_k} \sum_{k=1}^n \left[g\left(N_k\right)
+g\left(\eta_k N_k\right)-g\left((1-\eta_k)N_k\right)
\right].
\end{equation}
 where  $\overline{C_{E}}(\{{\cal N}_{\eta_k}\},n)$ denotes the upper bound on
$C_{E}(\{{\cal N}_{\eta_k}\},n,\epsilon)$.
Clearly $\lim_{n\to\infty} C_{E}(\{{\cal N}_{\eta_k}\},n,\epsilon)=C_E(\{ {\cal N}_{\eta_k} \})$.

The variance of the quantum relative entropy  $\overline{C_{E}}(\{{\cal N}_{\eta_k}\},n)$ can be obtained by means of
\eqref{re:34} as
\begin{equation}\label{VCE}
\frac{1}{n^2}\sum_{k=1}^n V\left({\cal N}_{A\to B}^k(\rho_{R A}) \| \rho_{R} \otimes {\cal N}_{A\to B}^k(\rho_{A})\right).
\end{equation}


\subsection{Unassisted classical capacity}

 Here we want to see how the capacity resulting from Theorem \ref{Ccoding} is approached over channel uses.

From Ref.~\cite{WRG15}, we know that
\begin{equation}\label{eq:Dclass}
D(\rho_{XB_k}\Vert \rho_X\otimes \rho_{B_k})=g(\eta_k N_k) \,,
\end{equation}
and
\begin{equation}\label{eq:Vcc}
V(\rho_{XB_k}\Vert \rho_X\otimes \rho_{B_k})=\eta_k N_k (\eta_k N_k+1)
\left[\log\left(\eta_k N_k+1\right)-\log\left(\eta_k N_k\right)\right]^2 \,.
\end{equation}
Then, according to Theorem \ref{CEAcoding} and \eqref{QRE} we now need to
maximize the quantity
\begin{equation}
\sum_{k=1}^n g(\eta_k N_k),
\end{equation}
with  respect to $N_k$.

Proceeding like in the previous section,
using a Lagrange multiplier $\beta$ and imposing \eqref{constraint}, we get
\begin{equation}
\left\{
\begin{array}{cl}
g'\left(N_k \right)-\beta &=0 \\ \\
\sum_{k=1}^n N_k&=nN
\end{array}\right..
\end{equation}
Solving this set of $n+1$ equations
allows us to find the $N_k$ and $\beta$ giving
\begin{equation}\label{CEn}
 n\, \overline{C}(\{ {\cal N}_{\eta_k}\},n) \equiv\max_{N_k} \sum_{k=1}^n \left[g\left(N_k\right)
\right].
\end{equation}
 where  $\overline{C}(\{{\cal N}_{\eta_k}\},n)$ denotes the upper bound on
$C(\{{\cal N}_{\eta_k}\},n,\epsilon)$.
Clearly $\lim_{n\to\infty} \overline{C}(\{ {\cal N}_{\eta_k}\},n)=C(\{ {\cal N}_{\eta_k} \})$.

The variance of the quantum relative entropy
 $\overline{C}(\{ {\cal N}_{\eta_k}\},n)$
can be obtained by means of
\eqref{eq:Vcc} as
\begin{equation}\label{VC}
\frac{1}{n^2}\sum_{k=1}^n V(\rho_{XB_k}\Vert \rho_X\otimes \rho_{B_k}).
\end{equation}


 \section{Examples}\label{sec:EX}

 We now apply the results
of Sec.\ref{sec:Gauss} to some specific cases study,
i.e. specific sequences of lossy channels.

\subsection{Example 1}\label{sec:sub1}

Consider
\begin{equation}\label{etai1}
\eta_k=\underline{\eta}+\overline{\eta} e^{-(k-1)^2/\Delta}, \qquad 0 < \underline{\eta},\overline{\eta} <\frac{1}{2}\, .
\end{equation}
After a transient (whose extension is determined by $\Delta$) the channel reaches a transmissivity
$\underline{\eta}$ (see Fig.\ref{fig1}).  The distribution of input energy shows a similar behavior to transmissivity (see Fig.\ref{fig1}). Note however that the sequence  of input energies
$\{N_k^{(n)}\}_k$ depends on  the number $n$ of channel uses.

\begin{figure}[H]
	\centering
   \includegraphics[width=0.4\textwidth]{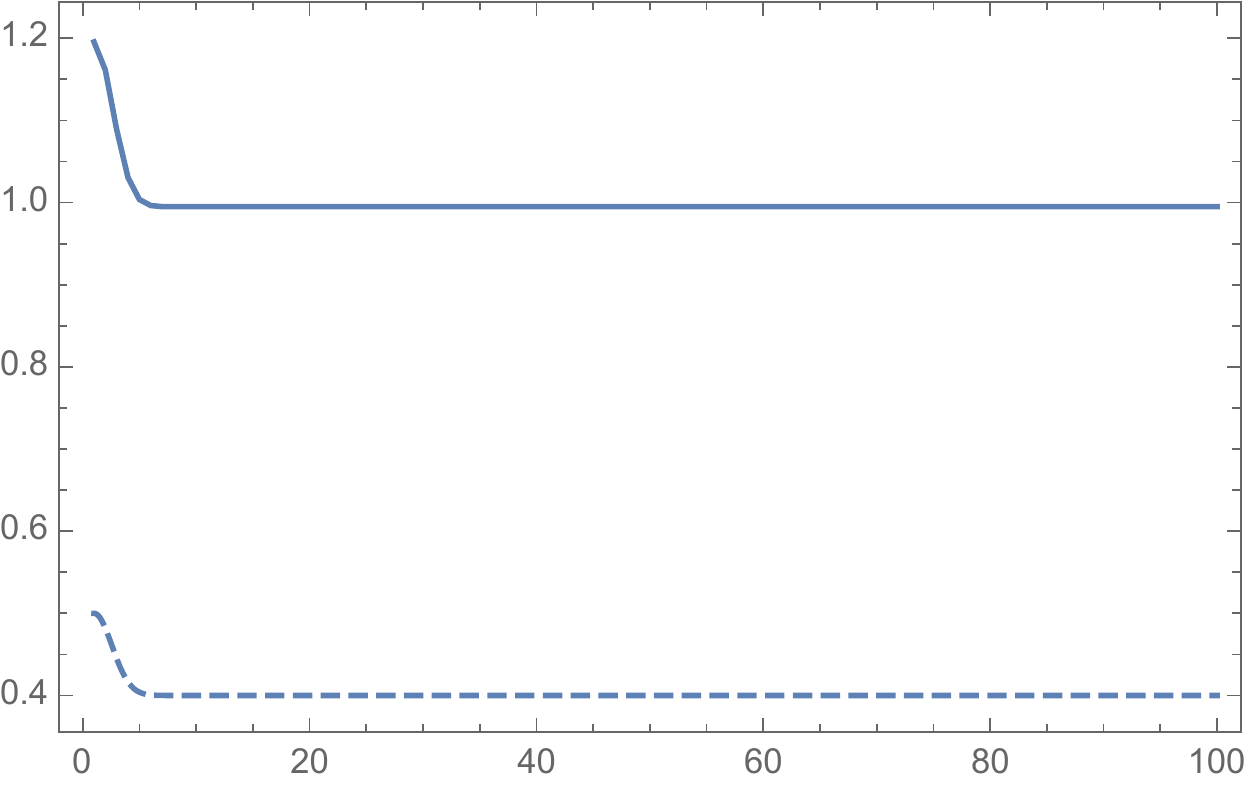} \quad
      \includegraphics[width=0.4\textwidth]{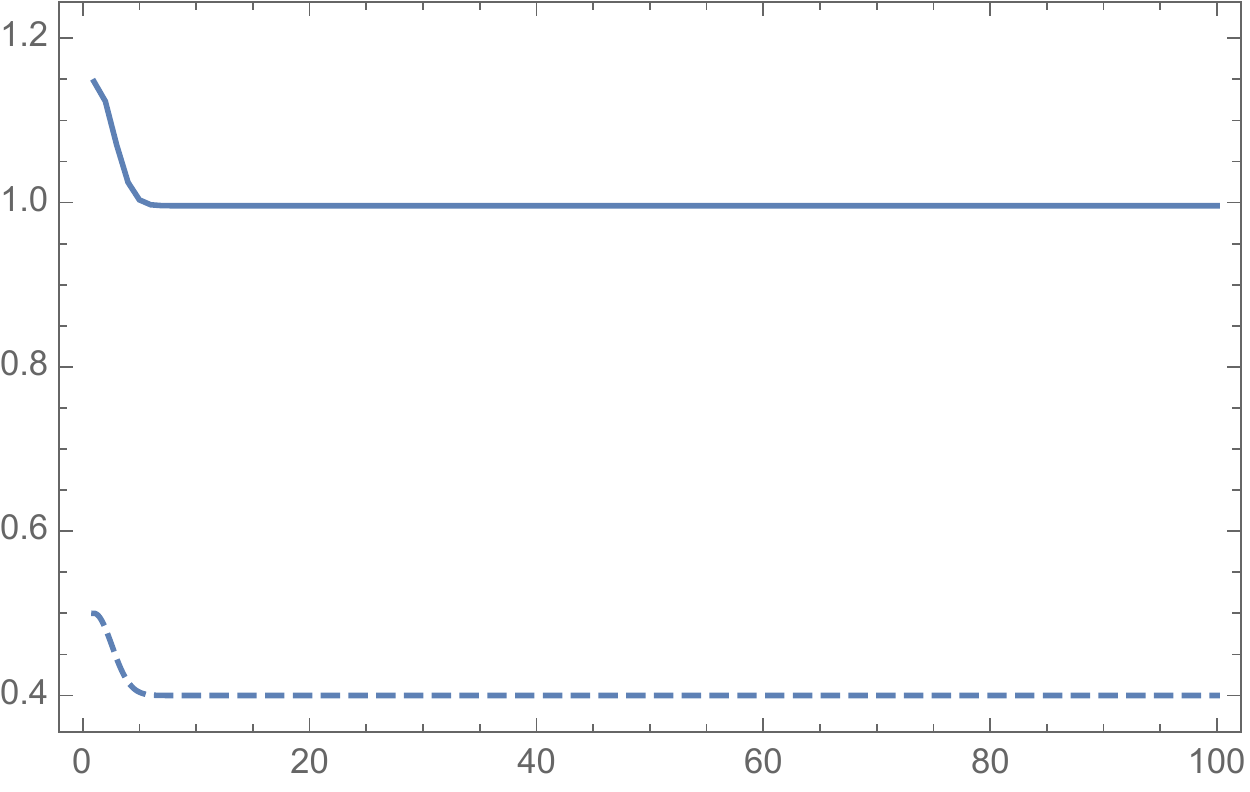}
             	\caption Quantities $\eta_k$ (dashed line) and  $N_k^{(n)}$ (solid line) vs $k$ for $n=100$. On the left (resp. right) is the case for entanglement assisted (resp. unassisted) classical communication.
	It is $\Delta=5$. The values of other parameters are $\underline{\eta}=0.4,\overline{\eta}=0.1,N=1$.
	\label{fig1}
\end{figure}

For this example the capacity
 $C_{E}(\{ {\cal N}_{\eta_k}\})$ can be guessed by simply computing
$\overline{C_E}(\{ {\cal N}_{\lim_{k\to\infty}\eta_k}\},n=1)$ (see Fig.\ref{fig2} left).
Analogous argument holds true for $C(\{ {\cal N}_{\eta_k}\})$, i.e. it can be guessed by simply computing $\overline{C}(\{ {\cal N}_{\lim_{k\to\infty}\eta_k}\},n=1)$ (see Fig.\ref{fig2} right).

\begin{figure}[H]
	\centering
	\includegraphics[width=0.4\textwidth]{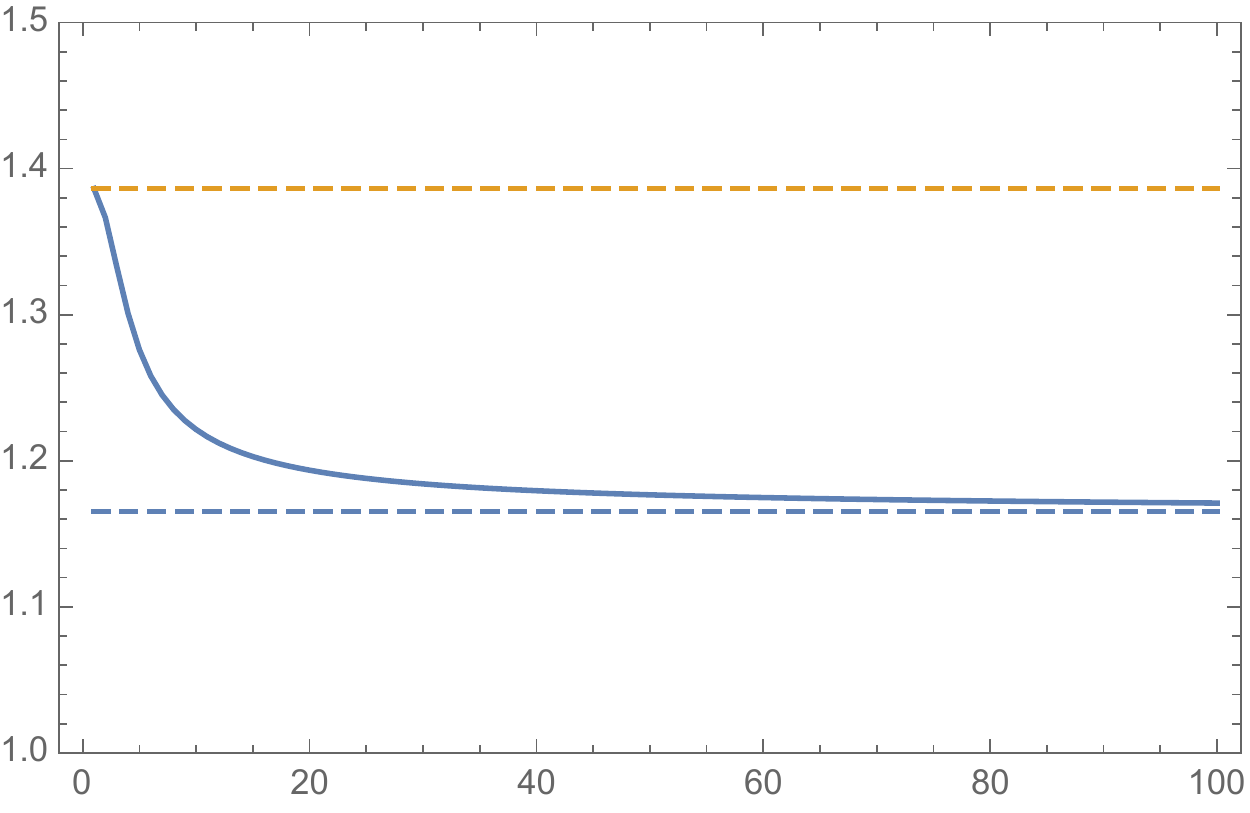} \quad
\includegraphics[width=0.4\textwidth]{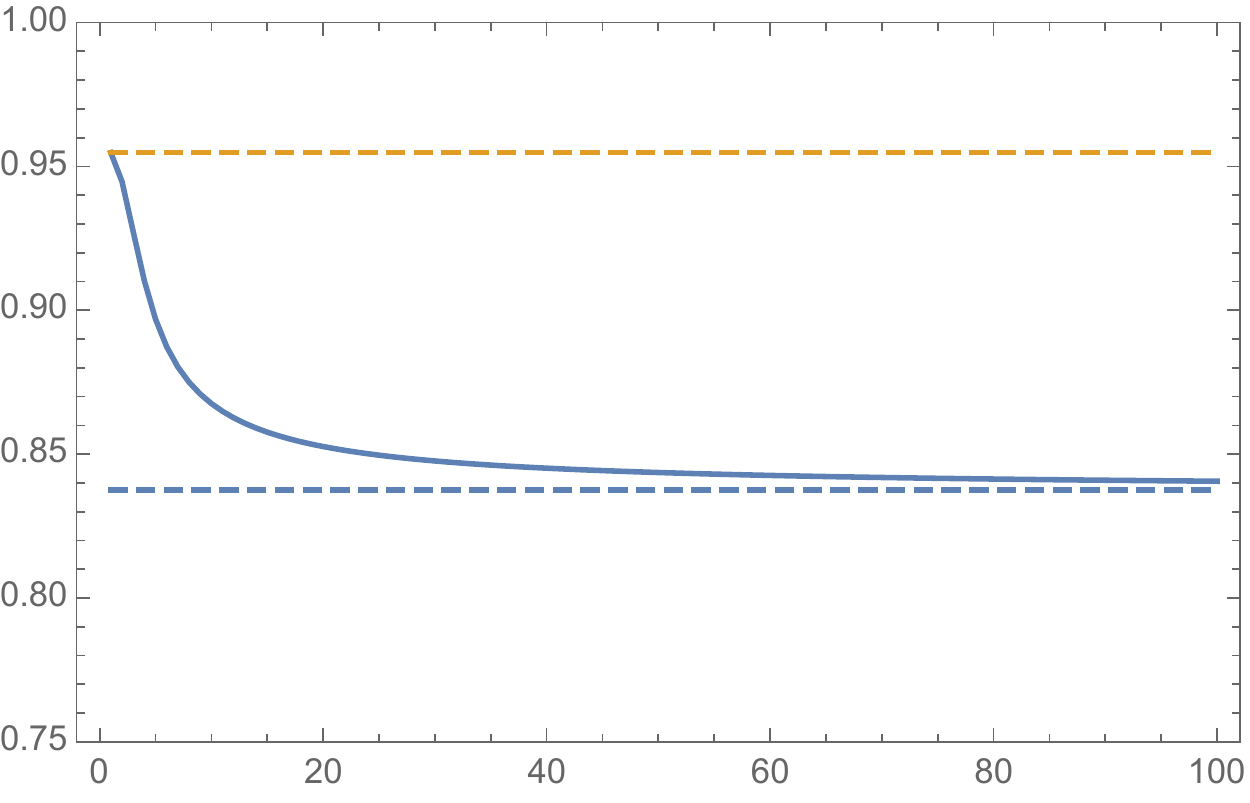}
	\caption{(Left) $\overline{C_E}(\{ {\cal N}_{\eta_k}\},n)$ vs $n$ for $\Delta=5$  (solid line). The bottom (resp. top) dashed line represents the capacity $g(N)+g(\underline{\eta}N)-g((1-\underline{\eta})N)$ (resp. $g(N)+g((\underline{\eta}+\overline{\eta})N)-g((1-(\underline{\eta}+\overline{\eta}))N)$).
	(Right)  $\overline{C}(\{ {\cal N}_{\eta_k}\},n)$ vs $n$ for $\Delta=5$  (solid line).  The bottom (resp. top) dashed line represents the capacity $g(\underline{\eta}N)$ (resp. $g((\underline{\eta}+\overline{\eta})N)$).
	The values of other parameters are $\underline{\eta}=0.4,\overline{\eta}=0.1, N=1$.}
	\label{fig2}
\end{figure}

In Fig.\ref{fig3} we report the variances \eqref{VCE} and \eqref{VC} as functions of $n$.  As one can see the variance of $\overline{C_E}(\{ {\cal N}_{\eta_k}\},n)$ converges faster than that of
$\overline{C}(\{ {\cal N}_{\eta_k}\},n)$.
\begin{figure}[H]
	\centering
	\includegraphics[width=0.4\textwidth]{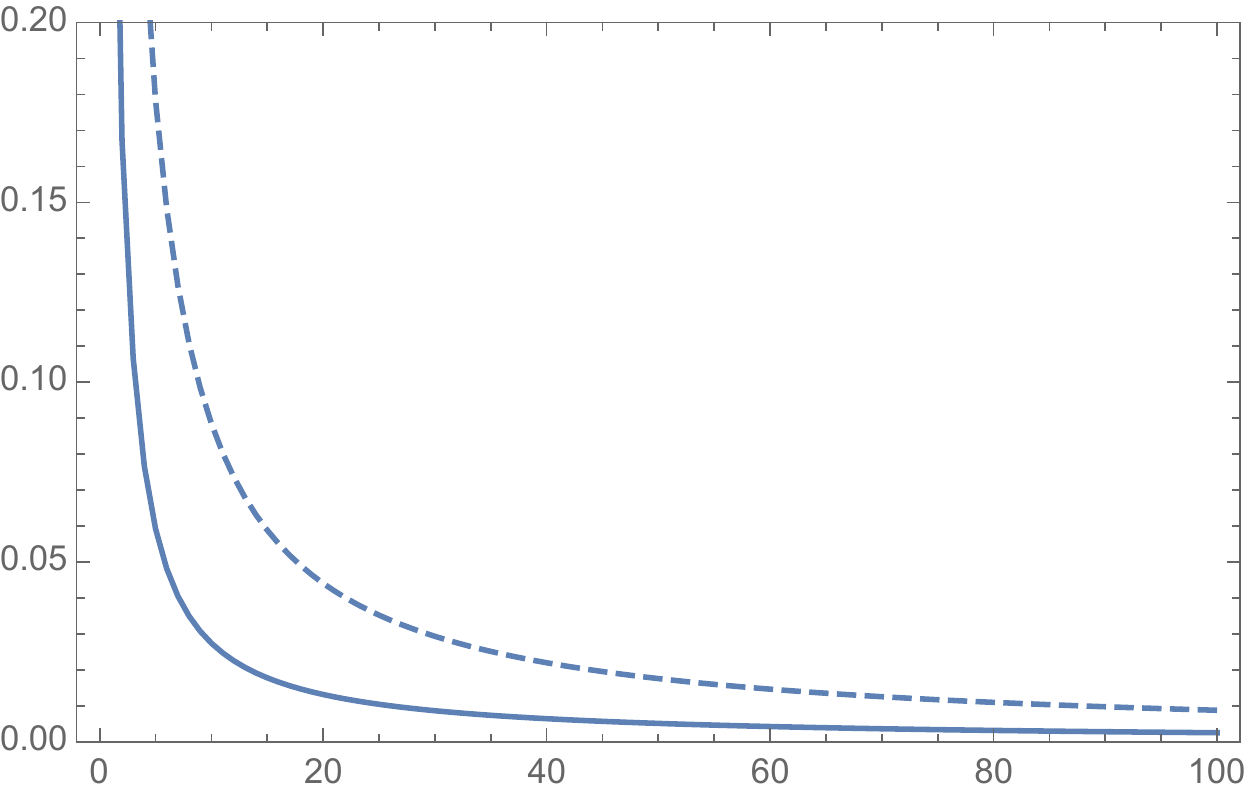}
	\caption{ Variance of $\overline{C_E}(\{ {\cal N}_{\eta_k}\},n)$ (solid line) and of
	$\overline{C}(\{ {\cal N}_{\eta_k}\},n)$ (dashed line) vs $n$ for $\Delta=5,\underline{\eta}=0.4,\overline{\eta}=0.1,N=1$.}
	\label{fig3}
\end{figure}


\subsection{Example 2}\label{sec:sub2}

Consider
\begin{equation}\label{etai2}
\eta_k=\underline{\eta}+\overline{\eta} \left|\sin\left(\frac{ k-1}{\Delta}+\frac{\pi}{2}\right) \right|, \qquad 0 < \underline{\eta},\overline{\eta} <\frac{1}{2}.
\end{equation}
In this case we have an oscillatory behavior of $\eta_k$ (whose frequency is determined by $\Delta$), as can be seen in Fig.\ref{fig4}. The distribution of input energy still follows the behavior of transmissivity.

\begin{figure}[H]
	\centering
	\includegraphics[width=0.4\textwidth]{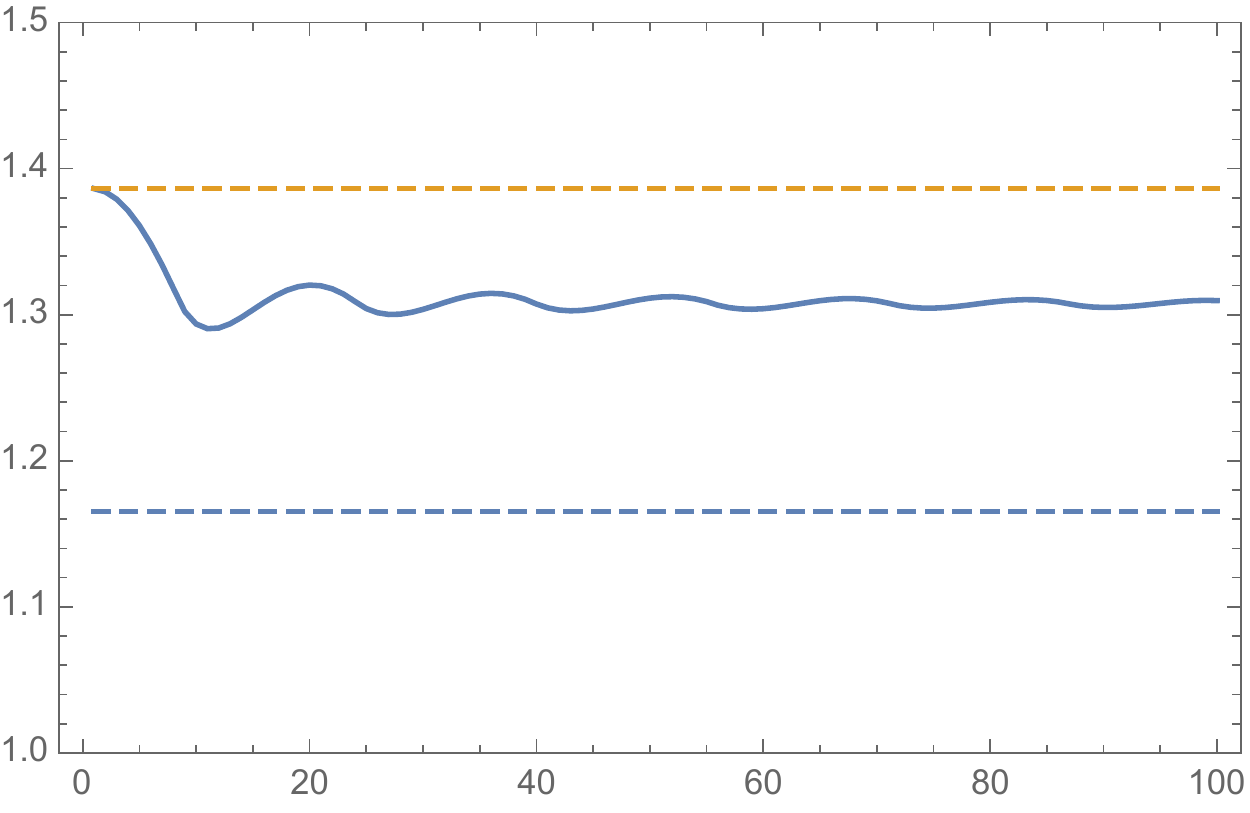} \quad
		\includegraphics[width=0.4\textwidth]{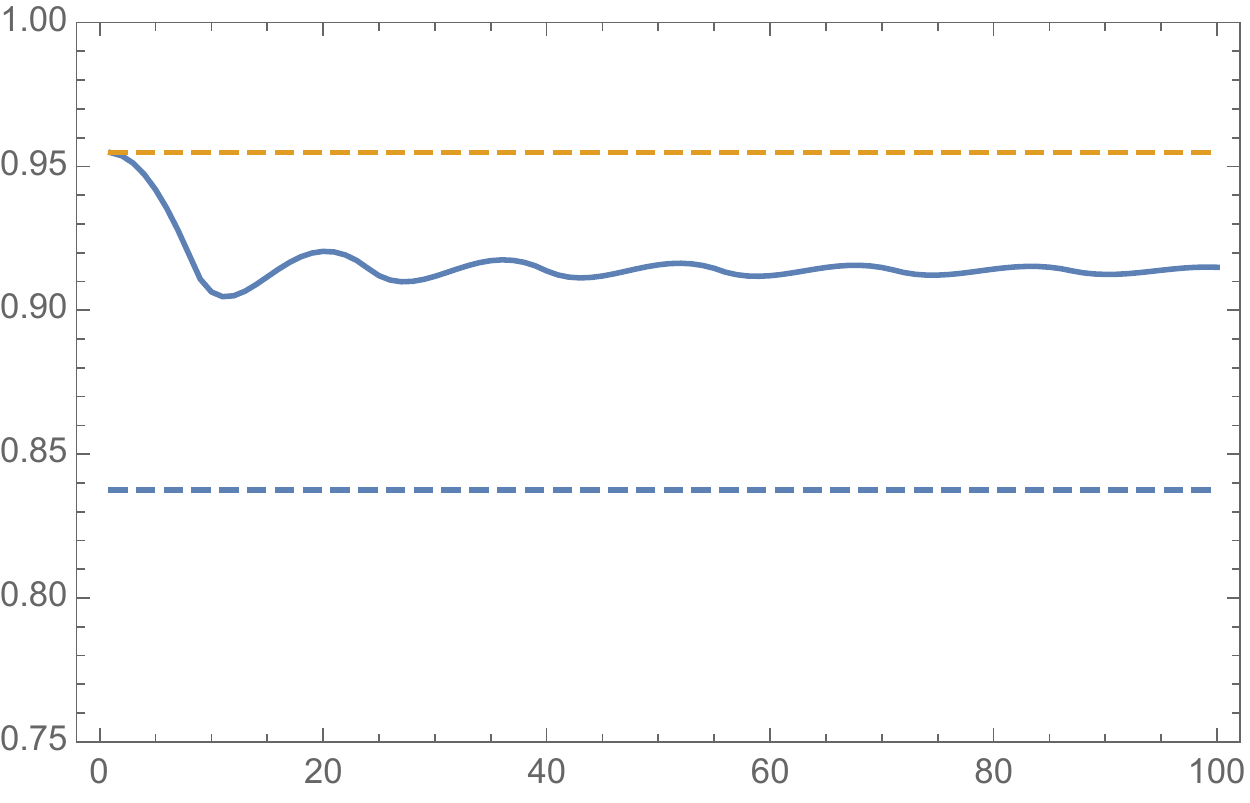}
	\caption{Quantities $\eta_k$ (dashed line) and  $N_k^{(n)}$ (solid line) vs $k$ for $n=100$. On the left (resp. right) is the case for entanglement assisted (resp. unassisted) classical communication.
	It is $\Delta=5$. The values of other parameters are $\underline{\eta}=0.4,\overline{\eta}=0.1,N=1$.
}
	\label{fig4}
\end{figure}

In this case the capacities cannot be guessed by the $\lim_{k\to\infty}\eta_k$, because this latter does not exist, however the  bound $\overline{C_E}(\{ {\cal N}_{\eta_k}\},n)$, after transient oscillations depending on $\Delta$, converges to a well definite value for $n \to \infty$ (see Fig.\ref{fig5} left). The same happens for the  bound $\overline{C}(\{ {\cal N}_{\eta_k}\},n)$ (see Fig.\ref{fig5} right).

\begin{figure}[H]
	\centering
	\includegraphics[width=0.4\textwidth]{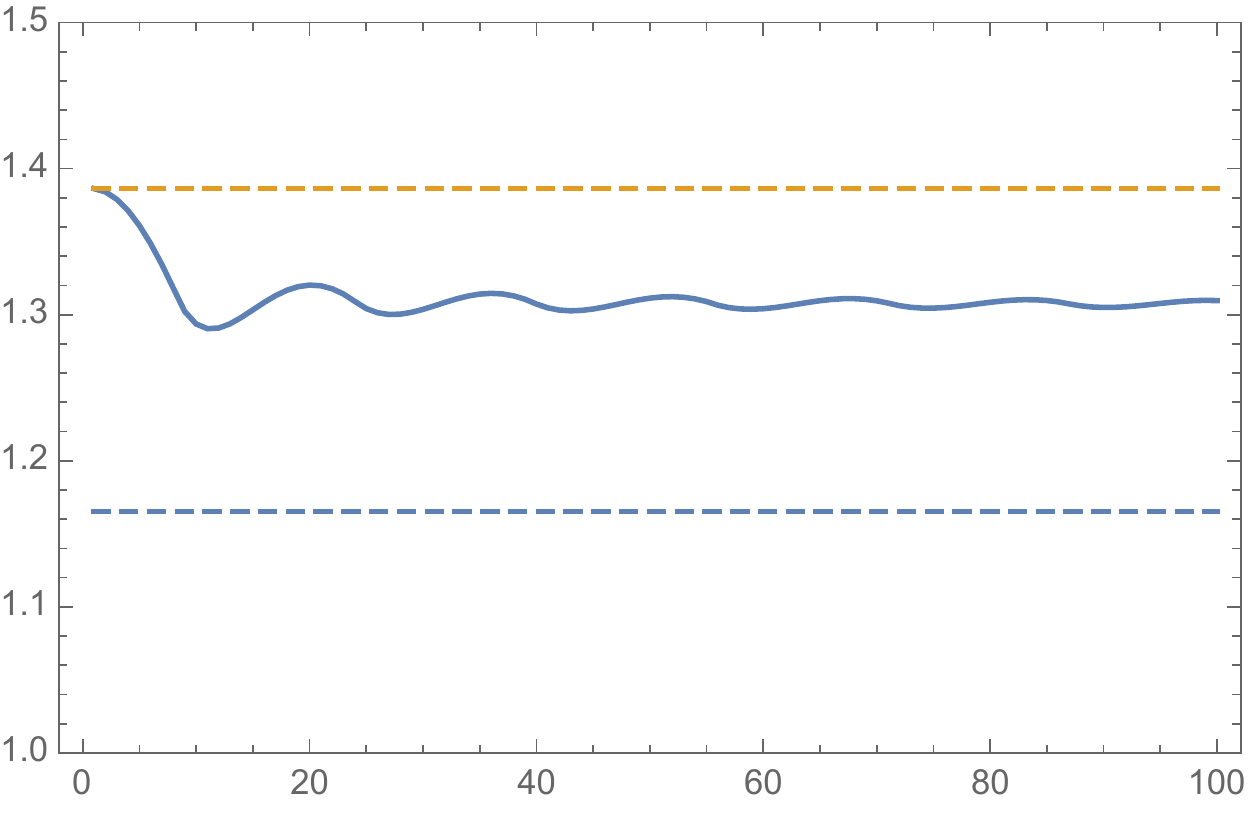} \quad
\includegraphics[width=0.4\textwidth]{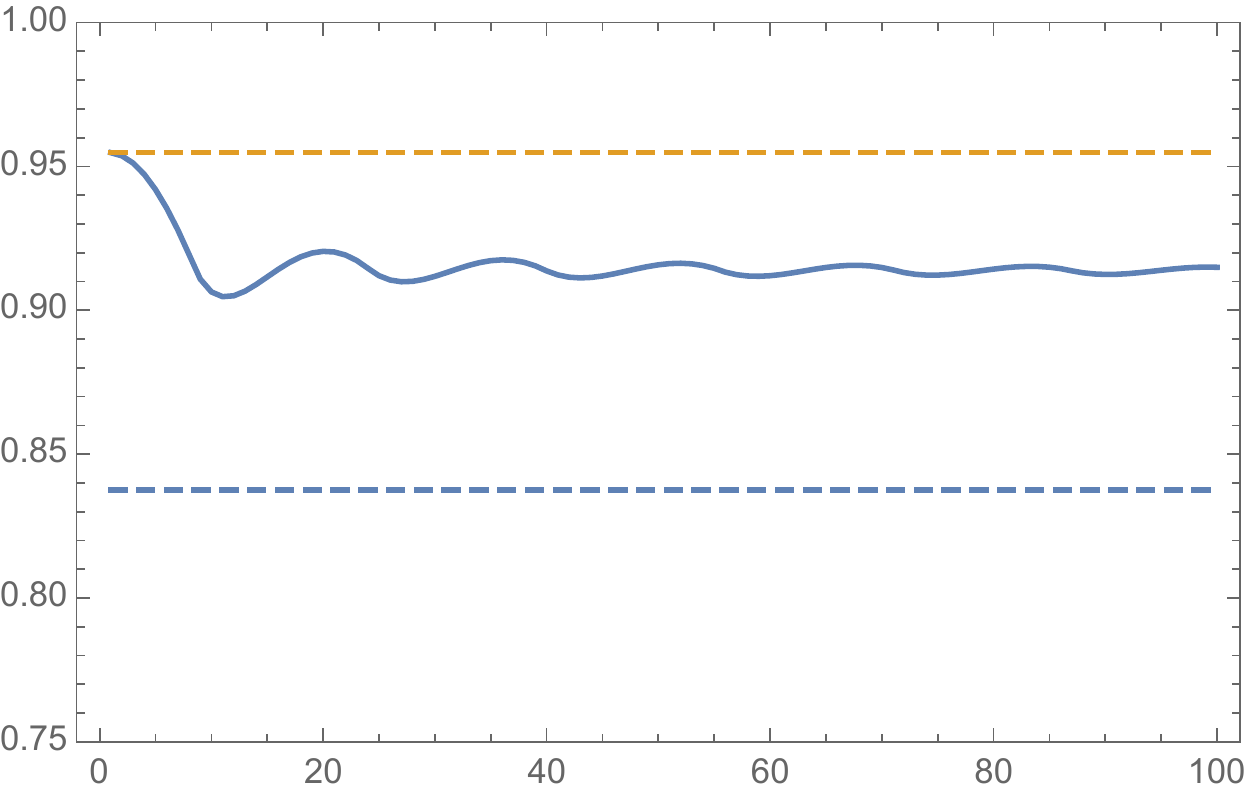}
		\caption{(Left)  $C_E(\{ {\cal N}_{\eta_k}\},n,\epsilon)$ vs $n$ for $\Delta=5$  (solid line). The bottom (resp. top) dashed line represents the capacity $g(N)+g(\underline{\eta}N)-g((1-\underline{\eta})N)$ (resp. $g(N)+g((\underline{\eta}+\overline{\eta})N)-g((1-(\underline{\eta}+\overline{\eta}))N)$).
	(Right)  $C_E(\{ {\cal N}_{\eta_k}\},n,\epsilon)$ vs $n$ for $\Delta=5$  (solid line).  The bottom (resp. top) dashed line represents the capacity $g(\underline{\eta}N)$ (resp. $g((\underline{\eta}+\overline{\eta})N)$).
	The values of other parameters are $\underline{\eta}=0.4,\overline{\eta}=0.1, N=1$.}
	\label{fig5}
\end{figure}

 As for what concerns the variance of  $\overline{C_E}(\{ {\cal N}_{\eta_k}\},n)$
 and of $\overline{C}(\{ {\cal N}_{\eta_k}\},n)$, the behavior is quite similar to that of Fig.\ref{fig3}.


\section{Conclusion}\label{sec:conclu}

In conclusion, we studied quantum channels that vary from one to another use in
 a deterministic way.
 To analyze their ability in transmitting classical information
we resorted to a smoothed version of quantum relative entropy.
As  the first result, we derived  a generalization of the relation between smooth and standard quantum relative entropies to the case of tensor product of non-identical density operators, which is valid also for separable Hilbert spaces (Eq.\eqref{newDDrel}).
We then proved coding theorems  for the classical entanglement assisted and unassisted capacities
(Theorems \ref{CEAcoding} and \ref{Ccoding} respectively). For that, we used position-based coding and quantum union bound applied to sequential decoding.
The results were then adapted, with the help of input energy constraints, to continuous variable quantum channels, specifically lossy bosonic.
Finally, enlightening examples were put forward in this context. They show that only when the sequence of channels parameter has a well defined limit, the capacities can be easily evaluated.
The approach taken allowed us to evaluate the maximum transmission rate for any number
of channel uses and estimate the error.

The natural extension of this work would be
the study of time-varying channels in a non-deterministic way.
Quite generally channels that are selected over the uses
according to a probability distribution that is itself varying from one to another use,
thus generalizing the arbitrarily varying channel model.
We are confident that the mathematical tools developed here will be useful to this end.
In another direction, one can pursue the quantum capacity of the introduced sequences of channels, which however needs slightly different tools.

\section*{Acknowledgments}
The authors are grateful to Mark M. Wilde for useful discussions.

\appendix

\section{Second-order asymptotic}\label{ap:second}

In this Appendix we derive Eq.\eqref{newDDrel}, which represents a generalization of the relation between smooth and standard quantum relative entropies to the case of tensor product of non identical density operators, which is valid also for separable Hilbert spaces.  The proof of inequality $\geq$ in Eq.\eqref{newDDrel} is based on \cite{DPR15, li12, OSM2018}.


Let us assume $\rho_i, \sigma_i$, for $i=1, 2, \ldots$ be full rank density operators on Hilbert space $\mathcal{H}$. Let us consider their spectral decompositions as follows
\begin{equation}
\rho_i=\sum_{x_i} \lambda_{x_i} P_{x_i}\, ,
\end{equation}
and
\begin{equation}
\sigma_i=\sum_{y_i} \mu_{y_i} Q_{y_i}\, .
\end{equation}
We also define two distributions for any two density operators $\rho=\sum_x  \lambda_x  P_x$ and  $\sigma=\sum_x  \mu_x  Q_y$ as follows
\begin{equation}
P_{\rho, \sigma}=\lambda_{x}\operatorname{Tr}(P_x Q_y)\, ,
\end{equation}
and
\begin{equation}
Q_{\rho, \sigma}=\mu_{y} \operatorname{Tr}(P_x Q_y)\, .
\end{equation}
We then introduce
\begin{equation}
P_{\rho, \sigma}^n(x_1,\ldots, x_n)=\prod_i P_{\rho, \sigma}(x_i)\, ,
\end{equation}
and
\begin{equation}
Q_{\rho, \sigma}^n(x_1,\ldots, x_n)=\prod_i Q_{\rho, \sigma}(y_i)\, .
\end{equation}

\begin{lemma}\label{lem6}
Let  $\rho$ and $\sigma$ be two density operators acting on a separable Hilbert space $\mathcal{H}$. For a given number $L$, there exists a measurement operator $T_L$ ($0\leq T_L\leq I$ ) such that
\begin{equation}
\operatorname{Tr}(T_L \rho)\geq \operatorname{Pr}(Z\geq \log{L}),\quad \operatorname{Tr}(T_L \sigma)\leq \frac{1}{L}\, ,
\end{equation}
where $Z$ is a random variable defined by $Z\equiv\log P_{\rho,\sigma}(X)-\log Q_{\rho, \sigma}(X)$.
\end{lemma}

Let $\bar{Z}=\frac{1}{n}\sum_i Z_i$ be the average over $n$ independent but not identical random variables $Z_i=\log P_{\rho,\sigma}(X_i)-\log Q_{\rho, \sigma}(X_i)$.
 For each $i$, set
\begin{equation}
  \mu_i\equiv D(P_{\rho_i, \sigma_i}\|Q_{\rho_i, \sigma_i})=D(\rho_i\|\sigma_i), \quad s_i^2\equiv V(P_{\rho_i, \sigma_i}\|Q_{\rho_i, \sigma_i})=V(\rho_i\|\sigma_i),
\end{equation}
and
\begin{equation}
t_n\equiv \mathbb{E}\left((Z_n-\mu_i)^3    \right)= \mathbb{E}\left( |\log P_i(X_i)-\log Q_i(X_i)-D(P\| Q)|^3  \right)\, .
\end{equation}
Now, we use the following Theorem which is a generalization of the central limit Theorem for independent but not identical random variables.

\begin{theorem}\cite{Fel71}.
  Let the $\{ X_n \}$ be random variables such that
\begin{equation}
\mathbb{E}(X_n)=0, \qquad \mathbb{E}(X_n^2)=s_n^2, \qquad \mathbb{E}(|X_n|^3)=t_n\, ,
\end{equation}
Put
\begin{equation}
\tilde{s}_n^2=s_1^2+\cdots+s_n^2,\qquad \tilde{t}_n=t_1 +\cdots+t_n\, ,
\end{equation}
and denoted by $P^n$ the distribution of the normalized sum $(X_1+\cdots+X_n)/\tilde{s}_n$, then  under the following condition
\begin{equation}\label{eq:centralcondition}
  \lim_{n\to \infty}\frac{6\tilde{t}_n}{\tilde{s}_n^3}=  0\,,
\end{equation}
for all $x$ and $n$, we have
\begin{equation}
\left| P^n\left(\frac{X_1+\cdots+X_n}{\tilde{s}_n}\leq x\right) - \Phi(x)   \right|\leq \frac{6\tilde{t}_n}{\tilde{s}_n^3}\, ,
\end{equation}
where
\begin{equation}
  \Phi(x)\equiv\int_{-\infty}^x \frac{1}{\sqrt{2\pi}} e^{-y^2/2} dy\, .
\end{equation}
\end{theorem}


\begin{proposition}
Let $\rho_1, \rho_2, \cdots$ and $\sigma_1, \sigma_2,\cdots$ denote states acting on a separable Hilbert
space~$\mathcal{H}$. Suppose that $D(\rho_i\Vert\sigma_i),V(\rho_i\Vert
\sigma_i),T(\rho_i\Vert\sigma_i)<\infty$ and $V(\rho_i\Vert\sigma_i)>0$, for each $i=1, 2, 3, \cdots$. Suppose $n$ is
sufficiently large  such that $  \varepsilon - \frac{6 \sum_{i=0}^n \left[T(\rho_i \| \sigma_i)\right]}{\sqrt{\left[\sum_{i=0}^nV\left(\rho_i \| \sigma_i)\right)\right]^3}}\geq 0$. Then%
\begin{align}
D_{H}^{\varepsilon}\left( \bigotimes_{i=1}^n\rho_i \Big\Vert\bigotimes_{i=1}^n\sigma_i\right) &  =
\sum_{i=1}^n D(\rho_i\Vert\sigma_i)+\sqrt{\sum_{i=1}^n V(\rho_i\Vert\sigma_i)}\Phi^{-1}\!\left(
\varepsilon-\frac{6 \sum_{i=0}^n \left[T(\rho_i \| \sigma_i)\right]}{\sqrt{\left[\sum_{i=0}^nV\left(\rho_i \| \sigma_i)\right)\right]^3}}\right)  \label{eq:2nd-order-expansion-final}\\
&  =\sum_{i=1}^n D(\rho_i\Vert\sigma_i)+\sqrt{\sum_{i=1}^n V(\rho_i\Vert\sigma_i)}\Phi^{-1}\!\left(
\varepsilon\right)  +O(\log n).\label{eq:2nd-order-expansion-final2}
\end{align}
\end{proposition}


\begin{proof}
\noindent\textbf{(Part $\geq$)} Applying the Berry--Esseen theorem \cite{Fel71}
to the random sequence $Z_{1}-D(\rho_1
\Vert\sigma_1)$, \ldots, $Z_{n}-D(\rho_n\Vert\sigma_n)$,  we find that%
\begin{equation}
\left\vert \Pr\left\{ \frac{ \overline{Z}^{n}}{\sqrt{\left[\sum_{i=0}^nV\left(\rho_i \| \sigma_i)\right)\right]^3}}\leq
x\right\}  -\Phi(x)\right\vert \leq  \alpha(n),
\end{equation}
where
$\overline{Z}^{n}\equiv\frac{1}{n}\sum_{i=1}^{n}\left[  Z_{i}
-D(\rho_i\Vert\sigma_i)\right]  $
and
$\alpha(n)\equiv\frac{6 \sum_{i=0}^n \left[T(\rho_i \| \sigma_i)\right]}{\sqrt{\left[\sum_{i=0}^nV\left(\rho_i \| \sigma_i)\right)\right]^3}}$, which implies that%
\begin{equation}
\Pr\left\{  \sum_{i=1}^{n}Z_{i}\leq \sum_{i=1}^{n} D(\rho_i\Vert\sigma_i)+ n x \sqrt{\left[\sum_{i=0}^nV\left(\rho_i \| \sigma_i)\right)\right]^3}\right\}  \leq\Phi(x)+\alpha(n).
\end{equation}
Picking $x=\Phi^{-1}\!\left(  \varepsilon-\frac{6\cdot \tilde{t}_n}{\tilde{s}_n^3}\right)  $, this becomes%
\begin{equation}
\Pr\left\{  \sum_{i=1}^{n}Z_{i}\leq \sum_{i=1}^{n}D(\rho_i\Vert\sigma_i)+\sqrt{\sum_{i=1}^{n}V(\rho_i
\Vert\sigma_i)}\Phi^{-1}\!\left(  \varepsilon-\alpha(n)\right)  \right\}
\leq\varepsilon.
\end{equation}
Choosing $L$ such that
\begin{equation}
\log L=\sum_{i=1}^{n}D(\rho_i\Vert\sigma_i)+\sqrt{\sum_{i=1}^{n} V(\rho_i\Vert\sigma_i)}\Phi^{-1}\!\left(
\varepsilon-\alpha(n)\right),
\end{equation}
and applying Lemma~\ref{lem6}, we find that%
\begin{align}
\operatorname{Tr}\{T^{n}\bigotimes_{i=1}^n \rho_i\}   \geq\Pr\left\{  \sum_{i=1}%
^{n}Z_{i}\geq\log L\right\}
  =1-\Pr\left\{  \sum_{i=1}^{n}Z_{i}\leq\log L\right\}
  \geq1-\varepsilon,
\end{align}
while%
\begin{equation}
\operatorname{Tr}\{T^{n}\bigotimes_{i=1}^n \sigma_i\}   \leq\frac{1}{L}
  =\exp\left\{-\left[ \sum_{i=1}^{n}D(\rho_i\Vert\sigma_i)+\sqrt{\sum_{i=1}^{n} V(\rho_i\Vert\sigma_i)}\Phi^{-1}\!\left(
\varepsilon-\alpha(n)\right)  \right] \right\}.
\end{equation}
This implies that
\begin{equation}
-\log\operatorname{Tr}\left\{T^{n}\bigotimes_{i=1}^n \sigma_i\right\}\geq \sum_{i=1}^{n}D(\rho_i\Vert\sigma_i)+\sqrt{\sum_{i=1}^{n} V(\rho_i\Vert\sigma_i)}\Phi^{-1}\!\left(
\varepsilon-\alpha(n)\right)  .
\end{equation}
Since $D_{H}^{\varepsilon}(\bigotimes_{i=1}^n \rho_i\Vert\bigotimes_{i=1}^n \sigma_i)$ involves
an optimization over all possible measurement operators $T^{n}$ satisfying
$\operatorname{Tr}\{T^{n}\bigotimes_{i=1}^n \rho_i\}\geq1-\varepsilon$, we conclude that the
bound $\geq$ in \eqref{eq:2nd-order-expansion-final} holds true. The equality
\eqref{eq:2nd-order-expansion-final2} follows from expanding $\Phi^{-1}$ at the
point $\varepsilon$ using Lagrange's mean value theorem.

\bigskip

\noindent\textbf{(Part $\leq$)} We use the following

\begin{theorem}~\cite{JOPS2012}
Let $\rho$ and $\sigma$ be  density operators acting on a separable Hilbert space $\mathcal{H}$, let $T$ be a measurement operator acting on $\mathcal{H}$ and such that $0\leq T \leq I$, and let $\nu, \theta\in\mathds{R}$. Then
\begin{equation}
  e^{-\theta}\operatorname{Tr}\{(I-T)\rho +\operatorname{Tr}\{T\sigma \} \geq \frac{e^{-\eta}}{1+e^{\nu-\theta}} \operatorname{Pr}\{ X\leq \nu\}  \,,
  \end{equation}
  where $X$ is a random variable taking values $\log(\lambda_x/\mu_y)$ with probability
  $\lambda_x \operatorname{Tr}(P_x Q_y)$.
  \end{theorem}
 Then, the proof closely follows the proof of Proposition $2$ in ~\cite{KW17a}, which is based on ~\cite{DPR15}. Choosing
\begin{equation}
 \nu_n= \sum_{i=1}^{n}D(\rho_i\Vert\sigma_i)+\sqrt{\sum_{i=1}^{n} V(\rho_i\Vert\sigma_i)}\Phi^{-1}\!\left(
\varepsilon+\frac{2}{\sqrt{n}}+\alpha(n)\right) \,,
\end{equation}
and $\theta_n=\nu_n+\frac{1}{2}\log n$,
we get
\begin{equation}
\operatorname{Tr}\{T^n \otimes_{i=1}^n \sigma_i  \}\geq \left\{e^{-\sum_{i=1}^{n}D(\rho_i\Vert \sigma_i)-\sqrt{\sum_{i=1}^{n}V(\rho_i\Vert \sigma_i)   }\Phi^{-1}(\varepsilon+2n^{-1/2}+\alpha(n))-\frac{1}{2}\log{n} }    \right\}\left( \frac{1}{1+n^{-1/2}}  \right)\,,
\end{equation}
where $\operatorname{Tr}\left\{(I^{\otimes n}-T^n)\otimes_{i=1}^n \rho_i\right\}\leq \varepsilon$.
 In this way we find
\begin{equation}
-\log{\operatorname{Tr}\{T^n \otimes_{i=1}^n \sigma_i  \}}  \leq  \sum_{i=1}^{n}D(\rho_i\Vert \sigma_i)
+\sqrt{\sum_{i=1}^{n}V(\rho_i\Vert \sigma_i)   }\Phi^{-1}(\varepsilon+2n^{-1/2}+\alpha(n))+\frac{1}{2}\log{n}     -\log{ \frac{1}{1+n^{-1/2}}}\,.
\end{equation}
\hfill
\end{proof}


%
\end{document}